\documentclass[12pt]{article}
\usepackage{cite}
\usepackage{amssymb}
\usepackage{amsmath}
\usepackage{amsthm}
\usepackage{fullpage}
\newtheorem{theorem}{Theorem}[section]
\newtheorem{corollary}{Corollary}[theorem]
\newtheorem{lemma}[theorem]{Lemma}

\begin{document}

\title{Access Balancing in Storage Systems by Labeling Partial Steiner Systems}

\author{Yeow~Meng~Chee\thanks{Yeow Meng Chee is with the Department of Industrial Systems Engineering and Management, National University of Singapore, Singapore}%
, 
Charles~J.~Colbourn\thanks{Charles J. Colbourn is with the School of Computing, Informatics, and Decision Systems Engineering, Arizona State University, Tempe AZ, USA}%
, 
Hoang~Dau\thanks{Hoang Dau is with Computer Science and Information Technology, RMIT University, Melbourne, Australia}%
,  
Ryan~Gabrys\thanks{Ryan Gabrys is with Electrical and Computer Engineering, University of Illinois at Urbana-Champaign, Urbana, IL, USA}%
, \\
Alan~C.H.~Ling\thanks{Alan C.H. Ling is with the Department of Computer Science, University of Vermont, Burlington, VT, USA}
,  
Dylan~Lusi\thanks{Dylan Lusi is with the School of Computing, Informatics, and Decision Systems Engineering, Arizona State University, Tempe AZ, USA}, and
Olgica~Milenkovic\thanks{ Olgica Milenkovic is with Electrical and Computer Engineering, University of Illinois at Urbana-Champaign, Urbana, IL, USA} 
}

\maketitle

\begin{abstract}
Storage archtectures ranging from minimum bandwidth regenerating encoded distributed storage systems to declustered-parity RAIDs can be designed using dense partial Steiner systems in order to support fast reads, writes, and recovery of failed storage units. 
In order to ensure good performance, popularities of the data items should be taken into account and the frequencies of accesses to the storage units made as uniform as possible. A proposed combinatorial model ranks items by popularity and assigns data items to elements in a dense partial Steiner system so that the sums of ranks of the elements in each block are as equal as possible.  
By developing necessary conditions in terms of independent sets, we demonstrate that certain Steiner systems must have a much larger difference between  the largest and smallest block sums than is dictated by an elementary lower bound.
In contrast, we also show that certain dense partial $S(t,t+1,v)$ designs can be labeled to realize the elementary lower bound. 
Furthermore, we prove that for every admissible order $v$, there is a Steiner triple system ($S(2,3,v)$) whose largest difference in block sums is within an additive constant of the lower bound.
\end{abstract}


\section{Introduction}

Distributed storage systems~\cite{el2010fractional,silberstein2015optimal}, systems for batch coding~\cite{silberstein2016optimal}, and multiserver private information retrieval systems~\cite{fazeli2015codes} have each employed combinatorial designs for data placement, so that  elements of the design are associated with data items and blocks with storage units.  
In these contexts, the most common types of designs employed are $t$-designs and $t$-packings.
A \textit{$t$-$(v,k,\lambda)$ packing} is a pair $(X, \mathcal{B})$, where $X$, the point set, is a $v$-set and and $\mathcal{B}$ is a collection of $k$-subsets ({\em blocks})  of $X$ such that every $t$-subset of $X$ is contained in at most $\lambda$ blocks. 
The packing is a \textit{$t$-$(v,k,\lambda)$ design} when every $t$-subset of $X$ is a subset of exactly $\lambda$ blocks.
A $t$-$(v,k,1)$ design is a \textit{Steiner system}, denoted by S$(t,k,v)$. 
A $2$-$(v,3,1)$ design is  a \textit{Steiner triple system of order v}, denoted by STS$(v)$.
When $\lambda = 1$, a $t$-$(v,k,1)$ packing is also referred to as a {\em partial $S(t,k,v)$} or {\em partial Steiner system}.

When data items are of the same size, and data is placed on storage units using a $t$-design, placement of data is uniform across the storage units.  
Indeed in  $t$-$(v,k,\lambda)$ design, every point appears in exactly $r = \frac{\lambda \binom{v}{t}}{\binom{k}{t}}$ blocks; this is the {\em replication number} of the design.
In order to understand why Steiner systems can be employed in data placement, we outline some examples.
Large-scale distributed storage systems (DSS) must address potential  loss of storage units, while not losing data. 
One solution is to replicate each data item and distribute these replicas among multiple storage nodes; systems such as the Hadoop Distributed File System and the Google File System employ this strategy~\cite{cidon2013copysets}. One can further mitigate information loss by sensibly organizing the data. 
For example, exact Minimum Bandwidth Regenerating (MBR) codes~\cite{el2010fractional} consist of two subcodes, an outer MDS code along with an inner \textit{fractional repetition code} (FRC) that support redundancy and repairability, respectively.  
To make this precise, an {\em $(n,k,d)$-DSS} with $k \leq d \leq n$ consists of $n$ storage nodes in which a read can be accomplished by access to $k$ nodes and a failed node recovered by access to $d$ nodes.  
A fractional repetition code $\mathcal{C}$ \cite{el2010fractional} with repetition degree $\rho$ for an $(n,k,d)$-DSS is a collection $\mathcal{C}$ of $n$ subsets $V_1,V_2,\dots, V_n$ of a set $V$, $|V| = v$, and of cardinality $d$ each, satisfying the condition that each element of $V$ belongs to exactly $\rho$ different sets in the collection. 
The \textit{rate} of the FRC is $\min_{I \subset [n], |I| = k} |\cup_{i \in I} V_i|$. 
To optimize the rate and ensure correct repetition and repair, we require that $|V_i \cap V_j| \leq 1$ whenever $i \neq j$.
When $\rho = \frac{v-1}{d-1}$, such an FRC is a Steiner $2$-$(v,d,1)$ design with replication number $\rho$, where the set of (coded) file chunks $V$ is the set of points and the set of storage nodes $\{V_1, \dots, V_n\}$ is the set of blocks of the design. 

Steiner systems also prove useful for applications needing both high data availability and throughput, such as transaction processing. 
The storage systems underlying these applications require uninterrupted operation, satisfying user requests for data even in the event of disk failure and repairing these failed disks, on-line, in parallel. Continuous operation alone is not sufficient, because  such systems cannot afford to suffer significant loss of  performance during disk failures. \textit{Declustered-parity RAIDS} (DPRAIDs) are designed to satisfy these requirements~\cite{chen1994raid,holland1992parity}. 
Like standard RAIDs (short for ``Redundant Arrays of Inexpensive Disks''), DPRAIDs handle disk failure by using parity-encoded redundancy, in which subsets of the stored data (called \textit{parity stripes}) are XORed together to store a single-error-correction code. 
Unlike standard RAIDs, however, all disks in the DPRAID cooperate in the reconstruction of all the data units on a single failed disk. 
One can represent a DPRAID as a $t$-$(v,k,\lambda)$ design $(X,\mathcal{B})$, with $X$ ($|X| = v$) being the set of disks in the array, and $\mathcal{B}$ being the set of all parity stripes, each of size $k$. 
Then each disk occurs in the same number $c$ of parity stripes, guaranteeing that the reconstruction effort is distributed evenly. 

Although designs arise naturally in balancing data placement, little attention has been paid to the relative popularity of the data items.
However, one can exploit popularity information in order to  improve the relative equality of access among the storage units.
Dau and Milenkovic~\cite{dau2018maxminsum} formulate a number of problems to address access balancing, by labeling the points of the underlying design.
In order to introduce their problems and results, we first present more definitions and known results concerning designs.  

Although storage systems handle ``hot'' (frequently accessed) and ``cold'' (infrequently accessed) data categories differently, they do not take the long-term popularity of the data items within each category into account, which may result in unbalanced access frequencies to the storage units.  
Access balancing can be achieved in part by selecting an appropriate packing or design, and by appropriate association of data items with elements of the packing or design. Dau and Milenkovic~\cite{dau2018maxminsum} propose a combinatorial model that ranks data items by popularity, and then strives to ensure that the sums of the ranks of the data elements in each block are not too small, not too large, or not too different from block to block. 
In~\S\ref{pointlabel} we summarize their model, state elementary  bounds on  various block sums, and provide a small but important improvement in the lower bound on the smallest possible difference among the block sums in a Steiner triple system.
In~\S\ref{indep} we establish a close connection between such block sums and the size of a maximum independent set of elements in the packing or design.  
For certain designs, this connection can be used to show that, no matter how data items are associated with the elements of the design, the block sums must be  far from the values dictated by the elementary bounds from~\S\ref{pointlabel}.
Indeed, in order to approach the elementary bounds, one must select designs or packings with very specific properties; we pursue this in~\S\ref{sec:mis}.
The described findings indicate the need to find specific $S(t,k,v)$ designs, or at least `dense' $t$-$(v,k,1)$ packings, to match the elementary bounds more closely.
In~\S\ref{dense}, we explore a construction of $t$-$(v,t+1,1)$ packings that asymptotically match the bounds and contain almost the same number of blocks as the full Steiner system $S(t,t+1,v)$.
Completion of the dense $t$-$(v,t+1,1)$ packings to a Steiner system $S(t,t+1,v)$ appears problematic for general $t$;  doing so without dramatically changing the block sums appears to be even more challenging. Nevertheless, in~\S\ref{sumsts}, we pursue this to establish, for every admissible order $v$,  the existence of a Steiner triple system of order $v$ whose difference in block sums is at most an additive constant more than the elementary lower bound.

\section{Point Labelings and Block Sums}\label{pointlabel}

Let $D = (V,{\mathcal B})$ be a $t$-$(v,k,\lambda)$ packing.  
A {\em point labeling} of $D$ is a bijection ${\sf rk} : V \mapsto \{0,\dots,v-1\}$; our interpretation is that {\sf rk} maps an element  to its rank by popularity.
The {\em reverse} $\overline{\sf rk}$ of a point labeling ${\sf rk}$ has $\overline{\sf rk}(i) = v-1-{\sf rk}(i)$ for each $i \in   \{0,\dots,v-1\}$.
With respect to a specific point labeling {\sf rk}, define ${\sf sum}(B,{\sf rk}) = \sum_{x\in B}{\sf rk}(x)$ when $B \in {\mathcal B}$.  
\renewcommand{\arraycolsep}{1pt}
Then define \[
\begin{array}{rcl} 
{\sf MinSum}(D,{\sf rk}) & = & \min({\sf sum}(B,{\sf rk}) : B \in {\mathcal B});\\
{\sf MaxSum}(D,{\sf rk}) & = & \max({\sf sum}(B,{\sf rk}) : B \in {\mathcal B});\\
{\sf DiffSum}(D,{\sf rk}) & = & {\sf MaxSum}(D,{\sf rk})-{\sf MinSum}(D,{\sf rk});\\
{\sf RatioSum}(D,{\sf rk}) & = & {\sf MaxSum}(D,{\sf rk})/{\sf MinSum}(D,{\sf rk}).\\
\end{array} \]

Following~\cite{dau2018maxminsum}, one primary objective is to choose point labelings to maximize the {\sf MinSum} and/or to minimize one of the other three.  
Access balancing is concerned primarily with minimizing the {\sf DiffSum} or {\sf RatioSum}; because of the similarity between these two entities we often focus on the {\sf DiffSum}.
Let ${\mathcal R}_D$ denote the set of all point labelings of $D$.
Noting that ${\sf MaxSum}(D,{\sf rk}) = k(v-1) - {\sf MinSum}(D,\overline{\sf rk}) $,
we define \[
\begin{array}{rcl} 
{\sf MinSum}(D) & = & \max({\sf MinSum}(D,{\sf rk}) : {\sf rk} \in {\mathcal R}_D);\\
{\sf MaxSum}(D) & = & k(v-1)  -  {\sf MinSum}(D);\\
{\sf DiffSum}(D) & = & \min({\sf DiffSum}(D,{\sf rk}) : {\sf rk}  \in {\mathcal R}_D);\\
{\sf RatioSum}(D) & = & \min({\sf RatioSum}(D,{\sf rk}) : {\sf rk}  \in {\mathcal R}_D).\\
\end{array} \]

If the storage system dictates the data layout and data items have the same size, we are free to permute the data items; this is captured by the selection of the point labeling {\sf rk}.  
If we are also free to choose the $t$-$(v,k,1)$ packing that determines the data layout, we may select a packing to improve the sum metrics defined.  
In order to capture this, let ${\mathcal D}_{t,k,v,b}$ denote the set of all $t$-$(v,k,1)$ packings having exactly $b$ blocks.
Then define \[
\begin{array}{rcl} 
{\sf MinSum}(t,k,v,b) & = & \max({\sf MinSum}(D) : D \in {\mathcal D}_{t,k,v,b});\\
{\sf MaxSum}(t,k,v,b) & = & k(v-1)  -  {\sf MinSum}(t,k,v,b);\\
{\sf DiffSum}(t,k,v,b) & = & \min({\sf DiffSum}(D) : D \in  {\mathcal D}_{t,k,v,b});\\
{\sf RatioSum}(t,k,v,b) & = & \min({\sf RatioSum}(D) : D \in {\mathcal D}_{t,k,v,b}).\\
\end{array} \]
When $b = \frac{\binom{v}{t}}{\binom{k}{t}}$, the packing is a Steiner system $S(t,k,v)$; in these cases we omit $b$ from the notation to get ${\sf MinSum}(t,k,v)$ and similarly for all other entities.

\begin{theorem}\label{basic}  {\rm \cite{dau2018maxminsum}} When $D$ is a Steiner system $S(t,k,v)$, 
\[ \begin{array}{ccccl} 
 {\sf MinSum}(D)  &\leq &{\sf MinSum}(t,k,v)&\leq& \frac{1}{2} (v(k-t+1) + k(t-2));\\
{\sf MaxSum}(D)&\geq& {\sf MaxSum}(t,k,v) & \geq & \frac{1}{2}(v(k+t-1)-kt);\\
{\sf DiffSum}(D)& \geq &{\sf DiffSum}(t,k,v) & \geq &(v-k)(t-1);\\
 {\sf RatioSum}(D)& \geq &{\sf RatioSum}(t,k,v) & \geq &\frac{v(k+t-1)-kt}{v(k-t+1) + k(t-2)}.
\end{array} \]
When $k=t+1$, 
${\sf MinSum}(D) \leq (v-1) + \binom{t}{2},$
${\sf MaxSum}(D) \geq  t(v-1) - \binom{t}{2},$ 
${\sf DiffSum}(D) \geq (t-1)(v-t-1),$ and 
${\sf RatioSum}(D) \geq {\sf RatioSum}(t,t+1,v) \geq \frac{t(v-1) - \binom{t}{2}}{(v-1) + \binom{t}{2}}.$

When in addition $t=2$ ($D$ is a Steiner triple system), the stronger bounds 
${\sf DiffSum}(D) \geq v$ and 
${\sf RatioSum}(D) \geq 2 $ hold.
\end{theorem}

Theorem~\ref{basic} provides bounds on the metrics across all Steiner systems $S(t,k,v)$ and all point labelings of them. 
In previous work, the focus has been on the {\sf MinSum} (or equivalently, by reversal, the {\sf MaxSum}). 
Dau and Milenkovic~\cite{dau2018maxminsum} use the Bose~\cite{bose1939construction} and Skolem~\cite{hanani1960note,skolem1959some} constructions of Steiner triple systems to establish the existence of an STS$(v)$ $D$ with ${\sf MinSum}(D)=v$, the largest possible by Theorem~\ref{basic} (These results are extended by Brummond~\cite{brummond2019kirkman} for Kirkman systems). They accomplish this by specifying a particular point labeling that meets the {\sf MinSum} bound, but unfortunately the labeling chosen yields a {\sf MaxSum} near $\frac{8}{3}v$, a {\sf DiffSum} near $\frac{5}{3}v$, and a {\sf RatioSum} near $\frac{8}{3}$, far from the bounds of $2v$,  $v$, and $2$, respectively. The reversal of this labeling yields a {\sf MinSum} far from optimal, the same {\sf DiffSum}, and a larger {\sf RatioSum}.

One might hope to improve the {\sf DiffSum} and  {\sf RatioSum} by choosing a different labeling or by choosing a different Steiner system $S(t,k,v)$.  
In Section~\ref{indep}, we show that certain $S(t,k,v)$s cannot meet {\sl any} of the bounds in Theorem~\ref{basic}.

\subsection{Improved bounds for STSs}
{

There is an STS$(7)$ with ${\sf MinSum} = 6$ and ${\sf MaxSum} = 13$ with blocks
$016$, $024$, $035$, $123$, $145$, $256$, and $346$ (here we write $abc$ for $\{a,b,c\}$).
There is an STS$(9)$ with ${\sf MinSum} = 9$ and ${\sf MaxSum} = 18$ with blocks
$018$, $027$, $036$, $045$, $126$, $135$, $147$, $234$, $258$, $378$, $468$, and $567$.
However, we  establish that these are the {\sl only} two Steiner triple systems with ${\sf DiffSum} = v$, and indeed the {\sl only} STS$(v)$ with ${\sf RatioSum}=2$ is the STS$(9)$. We first prove a useful lemma.

\begin{lemma}\label{maxxm1}
A $2$-$(x,3,1)$ packing on $\{0,\dots,x-1\}$ with {\sf MaxSum} $x-1$ has at most $\lfloor \phi(x)/3 \rfloor$ triples, where 
\[ \phi(x) =  \frac {x(x-1)}{4} -  \left \lfloor \frac{x}{6} \right \rfloor - 
\left \{ \begin{array}{lcccc} 
0 & \mbox{if} & x\equiv 0,1,4 & \pmod{6}\\
\frac{1}{2} & \mbox{if} & x\equiv 2,3 &\pmod{6}\\
 1& \mbox{if} & x\equiv 5 & \pmod{6}\\
 \end{array} \right . \]
\end{lemma}
\begin{proof}
We  determine an upper bound on the number of pairs that could appear in triples of the packing. 
In total there are $\binom{x}{2}$ pairs; of these, exactly $\lfloor \frac{x}{2} \rfloor$ have sum equal to $x-1$. 
For each pair $\{a,b\}$ with $a+b < x-1$, the pair $\{x-1-a,x-1-b\}$ has sum $2x-2-(a+b) > x-1$.  
It follows that the number of pairs with sum at most $x-1$ is $\frac{x^2}{4}$ when $x$ is even, and $\frac{x^2-1}{4}$ when $x$ is odd.
Not all of these can appear together in a packing, as follows.
Let $a \in \{0, \dots, \lfloor \frac{x-2}{3} \rfloor\}$. Consider the pairs $P_a = \{\{a,x-1-b\} : a \leq b \leq 2a\}$.
To place a pair of $P_a$ in a triple of sum at most $x-1$, the  third element must be from $\{0,\dots,a\}$, but it cannot be $a$.  
By the pigeonhole principle, at least one pair of $P_a$ cannot be in a triple of the packing, reducing the number of pairs available by  $\lfloor \frac{x+1}{3} \rfloor$.
Simple calculations now show that the number  $\phi(x)$ of pairs available is as given in the statement.
\end{proof}

\begin{theorem}\label{diffp1}
Let $D$ be a Steiner triple system of order $v\geq 13$.  Then ${\sf DiffSum}(D) \geq v+1$ and ${\sf RatioSum}(D) > 2$.
\end{theorem}
\begin{proof}
Let $v=2x+1$, and consider an STS$(2x+1)$ $D$ on elements $\{0,\dots,2x\}$, noting that  $x \geq 6$.  
Partition $\{0,\dots,2x\}$ into three classes $V_0 = \{0\}$, $V_s = \{1,\dots, x\}$, and $V_\ell = \{x+1,\dots,2x\}$.
Suppose to the contrary that ${\sf DiffSum}(D) = v$.  
Without loss of generality, ${\sf MinSum}(D) \in \{v-1, v\}$, for otherwise we can apply the argument to the reversal of $D$.
Let $m={\sf MinSum}(D)$  and $M= {\sf MaxSum}(D)$.
Because  $m \geq v-1$, all triples containing 0 contain one element of $V_s$ and one of $V_\ell$, as follows.
Consider the pair $\{0,w\}$ with $w \in V_s$.  The third element $y$ completing its triple satisfies $y \geq 2x-w  \geq 2x-x = x$.
Now $y \neq x$ because $\{0,x,x\}$ cannot be a triple, so
 $y > x$,  and hence $y \in V_\ell$.
 This accounts for all triples involving 0.
 
Call a pair {\em mixed} if it contains an element of $V_s$ and one from $V_\ell$, {\em pure} otherwise.
Similarly a triple is  {\em pure} if it lies entirely on $V_s$ or $V_\ell$, {\em mixed} when it has two from one and one from the other.
The number of mixed triples can be calculated as follows.
There are $x(x-1)$ mixed pairs not contained in triples containing 0, and each must be contained in a mixed triple.  
Because each mixed triple contains two mixed pairs, there are exactly $\frac{1}{2} x (x-1)$ mixed triples.
Each mixed triple covers one pure pair.
Hence the number of pure pairs to be covered by pure triples is $x(x-1) -\frac{1}{2} x (x-1) = \frac{1}{2} x (x-1)$, and there are $\frac{1}{6} x (x-1)$ pure triples.  

Form a collection ${\mathcal D}_s$ of triples on $\{0,\dots, x-1\}$ by including $\{a,b,c\}$ whenever $\{a+1,b+1,c+1\}$ is a pure triple on $V_s$, so that ${\mathcal D}_s$ contains  triples {\sl each having sum at least $m-3$}.
The reversal ${\mathcal E}_s$ then has each sum {\sl at most $3x-m$}.
Form a collection ${\mathcal D}_\ell$ of triples on $\{0,\dots, x-1\}$ by including $\{a,b,c\}$ whenever $\{2x-a,2x-b,2x-c\}$ is a pure triple on $V_\ell$, so that ${\mathcal D}_\ell$ contains triples {\sl each having sum at least $6x-M$}. 
The reversal ${\mathcal E}_\ell$ then has each sum {\sl at most $M-3x-3$}.

\noindent {\bf Case 1.}  ${\sf MinSum}(D) =v$ and hence ${\sf Maxsum}(D) = 2v$.     
 Then ${\mathcal E}_s$ and  ${\mathcal E}_\ell$ both have maximum sum at most $x-1$.
Applying Lemma~\ref{maxxm1} to ${\mathcal E}_s$ and to ${\mathcal E}_\ell$,  $D$ can contain at most $2\phi(x)$ pairs in pure triples, but $2\phi(x) < \frac{1}{2}x(x-1)$, which yields the contradiction. (Only when $v=9$ would there be  no contradiction.)

\noindent {\bf Case 2.}  ${\sf MinSum}(D) =v-1$ and hence ${\sf Maxsum}(D) = 2v-1$. Then  ${\mathcal E}_s$ has maximum sum $x$ and  ${\mathcal E}_\ell$  has maximum sum $x-2$.  
Because no  edge involving element $x-1$ can appear in a triple of  ${\mathcal E}_\ell$, the number of pairs covered by triples is at most $\phi(x-1)$ by 
Lemma~\ref{maxxm1}.  
By a similar argument, the number of pairs covered by triples of  ${\mathcal E}_s$ is at most $\phi(x+1)$ by Lemma \ref{maxxm1}. 
Because $\frac{(x-1)(x-2)}{4} + \frac{x(x+1)}{4} = \frac{x(x-1)}{2} + \frac{1}{2}$,  $\phi(x-1) + \phi(x+1) < \frac{1}{2}x(x-1)$, which yields the contradiction. (Only when $v=7$ would there be  no contradiction.)
\end{proof}

\section{Independent Sets}\label{indep}

Let $D = (V,{\mathcal B})$ be a $t$-$(v,k,\lambda)$ packing.  
An {\em independent set} in $D$ is a subset $X \subseteq V$ such that there is no $B \in \mathcal{B}$ with $B \subseteq X$. 
An independent set $I$ is {\em maximal} if there is no independent set $Y$ with $X \subset Y$, and {\em maximum} if there is no independent set $Y$ such that $|Y| > |X|$.
The \textit{independence number} of $D$, denoted $\alpha(D)$, is the size of a maximum independent set. 
There is a  close connection between the independence number of a packing and the quality of any of its labelings.
Prior to establishing this fact, we first improve the lower bounds on the {\sf DiffSum} and {\sf RatioSum} for Steiner triple systems.

\begin{lemma}
\label{IndepBoundsone}
A $t$-$(v,k,\lambda)$ packing $D$ has {\sf MinSum} at most $k \alpha(D) - \binom{k}{2}$, {\sf MaxSum} at least $k(v-1- \alpha(D)) + \binom{k}{2}$,  and {\sf DiffSum} at least $k(v+k-2- 2\alpha(D))$.
\end{lemma}
\begin{proof}
It suffices to prove the statement for {\sf MinSum}.  
No matter how $D$ is given a point labeling, on elements with ranks in $\{0,\dots,\alpha(D)\}$, there is a block.  The sum of this block is at most $\sum_{i=1}^{k} (\alpha(D)-(i-1))$.
\end{proof}

\begin{corollary}\label{oneindep}
Meeting the bound on {\sf MinSum} in Theorem \ref{basic}   for a $t$-$(v,k,1)$ packing $D$ requires that 
\[ \alpha(D)  \geq \frac{v(k-t+1)}{2k} +\frac{k+t-3}{2}.\]
\end{corollary}

For example, Corollary \ref{oneindep} states that a necessary condition for a partial Steiner triple system $D$ to have {\sf MinSum} equal to $v$ is that  $\alpha(D) \geq \frac{v}{3} +1$.

A single maximum  independent set yields an upper bound on the {\sf MinSum} or a lower bound on the {\sf MaxSum}, but cannot be used simultaneously 
for both. We prove this next.

Suppose that a $t$-$(v,k,\lambda)$ packing $D$ contains two {\sl disjoint}  independent sets of sizes  $\gamma_D$ and $\delta_D$, respectively, with $\gamma_D \geq \delta_D$.  
Set \[ \begin{array}{rcl}
\gamma_D' & =& \min \left (\gamma_D,\frac{v(k-t+1)}{2k} +\frac{k+t}{2}-1 \right ), \\
\delta_D' & =&  \min\left (\delta_D,\frac{v(k-t+1)}{2k} +\frac{k+t}{2}-1\right ). \end{array} \]
Two independent sets form a {\em maximum independent pair} when $\gamma_D' + \delta_D'$ 
is as large as possible.

\begin{lemma}
\label{IndepBoundstwo}
A $t$-$(v,k,\lambda)$ packing $D$ with a maximum independent pair of sizes $(\gamma_D,\delta_D)$ has  {\sf DiffSum} at least $k(v+k-2- \delta_D'-\gamma_D')$.
\end{lemma}
\begin{proof}
Form a point labeling of $D$ in which the smallest $x$ for which a block appears on $\{0,\dots,x-1\}$ also has a block on $\{v-x,\dots,v-1\}$; reverse the labeling if necessary to do this.
Suppose to the contrary that this labeling has {\sf DiffSum} less than  $k(v+k-2- \delta_D'-\gamma_D')$.
If no block appears on $\{ 0,  \dots, \frac{v(k-t+1)}{2k} +\frac{k+t}{2}-2\}$ set $c = \frac{v(k-t+1)}{2k} +\frac{k+t}{2}-1$.  Otherwise  let $c$ be the smallest value for which  a block appears on $\{0,\dots,c\}$, so that $\{0,\dots,c-1\}$ forms an independent set.  
Proceed similarly to select $d$ so that $\{v-d,\dots,v-1\}$ is an independent set.  
This labeling has {\sf MaxSum} at least $k(v-1- d) + \binom{k}{2}$ and {\sf MinSum} at most $k c - \binom{k}{2}$, based on Theorem~\ref{basic} and Lemma~\ref{IndepBoundsone}.
It follows that the {\sf DiffSum} for this labeling is at least $k(v+k-2- d-c)$, so we have $c+d > \gamma_D' + \delta_D'$.
This contradicts the requirement that the maximum independent pair have sizes $(\gamma_D,\delta_D)$.
\end{proof}

\begin{corollary}\label{twoindep}
Meeting the bound on {\sf DiffSum} in Theorem \ref{basic}  for a $t$-$(v,k,1)$ packing $D$ requires that $D$ have a maximum independent pair of sizes 
$( \frac{v(k-t+1)}{2k} +\frac{k+t}{2}-1,\frac{v(k-t+1)}{2k} +\frac{k+t}{2}-1)$.
\end{corollary}

For a pair of independent sets to be maximum in this context, there is no requirement that either be maximum, 
nor is it required that their combined size be as large as possible.
For a Steiner triple system, for example, Corollary \ref{twoindep} asks only for two disjoint independent sets, each of size at least $\frac{v}{3} +1$, for a combined size of $\frac{2v}{3} +2$.
Applying the $2v+1$ construction \cite{colbourn1999triple} twice to an STS$(v)$, we form an STS$(4v+3)$ having a maximum independent pair of sizes $(2v+2,v+1)$; despite the fact that the combined size is over $\frac{3}{4}$ of the size of the STS, such a pair could not lead to a {\sf DiffSum} that meets the bound of 
Theorem~\ref{basic}, because the second largest of the pair is too small.

Corollary~\ref{twoindep} gives a necessary condition, not a sufficient one.  Nevertheless, some bounds on the metrics can be stated.

\begin{lemma}\label{twosuff} 
When a $t$-$(v,k,1)$ packing $D$ has two disjoint independent sets of sizes $\alpha$ and $\beta$, there is a point labeling with ${\sf MinSum}(D) \geq \alpha + \binom{k}{2}$ and (for the same labeling) ${\sf MaxSum}(D) \leq k(v-1) - \beta-1 - \binom{k}{2}$, so ${\sf DiffSum}(D) \leq k(v-k) - \beta  - \alpha -1 $.
\end{lemma}
\begin{proof}
Any point labeling assigning labels $\{0,\dots,\alpha-1\}$ to the points of the independent set of size $\alpha$,
labels $\{v-\beta, \dots, v-1\}$ to the points of the independent set of size $\beta$,
and labels $\{\alpha, \dots,v-\beta-1\}$ to the remaining points, meets the stated bounds.
\end{proof}

A Steiner system $S(t,k,v)$ is {\em $2$-chromatic} if its elements can be partitioned into two classes, both being independent sets.  
When a 2-chromatic $S(3,4,v)$ $D$ exists (see, for example, \cite{DoyenV71,Ji2007,PhelpsRosa80}), Lemma \ref{twosuff} establishes that ${\sf DiffSum}(D) \leq 3v -17 $.  

Recall that Dau and Milenkovic~\cite{dau2018maxminsum} use the Bose and Skolem   constructions of Steiner triple systems.
In retrospect, this choice is well-justified because the Bose construction leads to maximum independent pairs of sizes $(\frac{v}{3} + 1, \frac{v}{3} +1)$ when $v \equiv 3 \pmod{6}$ and the Skolem construction leads to maximum independent pairs of sizes $(\frac{v+2}{3} + 1, \frac{v+2}{3} +1)$ when $v \equiv 1 \pmod{6}$. 

One must hence focus on Steiner triple systems, and on $t$-$(v,k,1)$ packings in general, having large sizes in maximum independent pairs.  
This choice is important, because not all such systems have even a single large independent set, as we explain next.  

\section{Small Maximum Independent Sets}\label{sec:mis}

Can one choose an arbitrary $t$-$(v,k,1)$ packing, and by cleverly choosing a point labeling optimize one or more of the sum metrics?
If not, how far from the bound of Theorem~\ref{basic} can the best point labeling be?
In order to discuss these questions, define \[ \begin{array}{ccl}
\alpha_{min}(t,k,v) &=& \min \{ \alpha(D): D \mbox{ is a }t\mbox{-}(v,k,1)\mbox{ packing}\}\mbox{, and}\\
\alpha_{min}^{\star}(t,k,v) &=& \min \{ \alpha(D): D \mbox{ is an } S(t,k,v) \}.
\end{array}
 \]
When an $S(t,k,v)$ exists, $\alpha_{min}(t,k,v) \leq \alpha_{min}^{\star}(t,k,v)$.  

Erd\H{o}s and Hajnal \cite{erdos} establish that $\alpha_{min}(2,3,v)\geq \lfloor \sqrt{2v}\rfloor$; indeed a simple greedy algorithm produces an independent set of this size.

A $t$-$(v,k,1)$ packing has each element in at most $\binom{v-1}{t-1}/\binom{k-1}{t-1} = \prod_{i=1}^{t-1} \frac{v-i}{k-i}$ blocks.  
Applying a result of Spencer \cite{spencer}  generalizing  Tur\'an's theorem for graphs, 
we obtain \[ \alpha_{min}(t,k,v) \geq c_k \frac{v}{\left( \prod_{i=1}^{t-1} \frac{v-i}{k-i} \right )^\frac{1}{k-1}} \]
for $c_k$ a constant independent of $v$.
For partial Steiner triple systems, this asserts that $\alpha_{min}(2,3,v)\geq c \cdot v \sqrt{2}/\sqrt{v-1}$, a small improvement on the Erd\H{o}s-Hajnal result.

State-of-the-art lower bounds  rely heavily on the following theorem, and  all differ only by constant factors. 
\begin{theorem} {\rm \cite{ajtai}}
Let $\kappa \geq 2$ be a fixed integer. 
Let $G$ be a $(\kappa+1)$-uniform hypergraph on $n$ vertices. 
Then there are constants $t_0(\kappa)$ and $n_0(\kappa,\tau)$ so that whenever
\begin{enumerate}
    \item $G$ is uncrowded (i.e., has no $2$-, $3$-, or $4$- cycles);
    \item the maximum degree $\Delta(G)$ satisfies $\Delta(G) \leq \tau^\kappa$ where $\tau \geq t_0(\kappa)$; and
    \item $n \geq n_0(\kappa,\tau)$, 
\end{enumerate}
one has that 
$$\alpha(G) \geq \frac{.98}{e} \cdot 10^{-5/\kappa} \cdot \frac{n}{\tau} \cdot (\ln \tau)^{1/\kappa}.$$
\end{theorem}

In order to apply this result to all $t$-$(v,k,1)$ packings, typically one selects a large subset of the blocks that are uncrowded.
The approach is used to establish the lower bound in the following, while the upper bound is shown using the Lov\'asz Local Lemma:

\begin{theorem}\label{GeneralBound} {\rm \cite{duke1995uncrowded,RodlS94}}
For fixed $k$ and $t$, there are absolute constants $c$ and $d$ for which
\[
    cv^{\frac{k-t}{k-1}}(\log v)^{\frac{1}{k-1}} \leq  \alpha_{min}(t,k,v) \leq dv^{\frac{k-t}{k-1}}(\log v)^{\frac{1}{k-1}}, \]
\end{theorem}

Variations in Theorem \ref{GeneralBound} have resulted in improvements in the constants; see  \cite{EustisV13,KostochkaMRT02,KostochkaMV14,TianL18}.

It is possible in principle that restricting to Steiner systems, rather than packings, one might observe different behaviour in the minima.  
However, Phelps and R\"odl \cite{phelps} establish that the bounds of Theorem \ref{GeneralBound} apply to Steiner triple systems, not just to partial ones; that is,
\[   c \sqrt {v \ln v } \leq  \alpha_{min}^{\star} (2,3,v) \leq   d \sqrt{v \ln v} \]
for absolute constants $c$ and $d$.  
Grable, Phelps and R\"odl \cite{GrablePR95} establish similar statements when $t \in \{2,3\}$ for all $k > t$.

For the applications intended, it is of interest to find independent sets of (at least) the size guaranteed efficiently.  For research in this vein, see \cite{algorithm,fundia}.
Of course, one wants to find a pair of disjoint maximum independent sets whose total size is as large as possible, but this is NP-complete even for 3-uniform hypergraphs \cite{Lovasz73}.
Remarkably, there is a polynomial time algorithm to determine whether an $S(3,4,v)$ contains two independent sets, each of size $v/2$ \cite{ColbournCPR82}, but the ideas used do not appear to generalize.

Nevertheless, the bounds on sizes of smallest maximum independent sets provides bounds on the best sum metrics one can hope to achieve.  
We state one such bound explicitly, showing that some Steiner systems have only point labelings far from the bounds of Theorem \ref{basic}.

\begin{theorem}
For infinitely many orders $v$, there exists an STS$(v)$ $D$ with ${\sf MinSum}(D) \leq 3c\sqrt{v \ln v} - 3$ and ${\sf MaxSum}(D) \geq 3v - 3c \sqrt{v \ln v}$, and hence ${\sf DiffSum}(D) \geq 3v - 6c\sqrt{v \ln v} + 3$.
\end{theorem}

Hence we must focus on specific Steiner systems or packings, if we are to obtain sum metrics at or near the basic bounds.

\section{Dense $t$-$(v,t+1,1)$ Packings}\label{dense}

We establish next that one can obtain metrics close to the optimal when $k=t+1$ for packings that contain all but a vanishingly small fraction of the blocks of an $S(t,t+1,v)$ as $v \rightarrow \infty$. The independent set requirements indicate that we must have a maximum independent pair having large sizes.  
To accomplish this, we partition all $(t+1)$-subsets of ${\mathbb Z}_v$ according to their sum modulo $v$, and choose one class of the partition to form the blocks of the packing. The basic strategy dates back at least a century  to Bussey~\cite{Bussey1914}, and perhaps earlier.

Prior to establishing our result, we note that this is not a mere theoretical curiosity; as Chen {\sl et al.} observe in~\cite{chen1994raid}, declustered-parity RAIDs do not in practice need to have their loads perfectly balanced. Hence, for practical reasons, one may  choose to omit some blocks from the design.

\begin{theorem}\label{mainpack}
Let $t$ and $v$ be integers with $v > {t+2 \choose 2} + {t+1 \choose 2}$ and $(v,t+1) = 1$. 
For each of the following statements, there exists a  $t$-$(v,t+1,1)$ packing $D$ on elements $\mathbb{Z}_v$ having
\begin{align*}
    \frac{{v \choose t+1}}{v} = \frac{v-t}{v} \frac{{v \choose t}}{{t+1 \choose t}}
\end{align*}
blocks.
\begin{enumerate}
    \item[(1)] ${\sf MinSum}(D) = v + \sigma$ and ${\sf MaxSum}(D) = tv + \sigma$ whenever $-{t+2 \choose 2} + 1 \leq \sigma < {t+1 \choose 2}$.
    
    \item[(2)] ${\sf MinSum}(D) = v + {t+1 \choose 2} - 1$.
    
    \item[(3)] ${\sf MaxSum}(D) =tv-{t+2 \choose 2} + 1$.
    
    \item[(4)] ${\sf DiffSum}(D) =(t-1)v$.
    
    \item[(5)] ${\sf RatioSum}(D) =\frac{tv+{t+1 \choose 2}-1}{v+{t+1 \choose 2}-1}$.
\end{enumerate}
\begin{proof}
It suffices to prove statement (1); the other results follow directly from it.

Partition all $(t+1)$-subsets of $\mathbb{Z}_v$ into $v$ classes $\{\mathcal{B}_\sigma : 0 \leq \sigma < v\}$ by placing set $S = \{x_1, \dots, x_{t+1}\}$ in class $\mathcal{B}_{\sigma}$ if and only if $\sigma \equiv \Sigma_{i=1}^{t+1} x_i\;(\bmod\;v)$. Because for any $t$-subset $T$ of $\mathbb{Z}_v$ and each $\sigma$ with $0 \leq \sigma < v$ there is a unique element $s$ for which $\sigma \equiv s + \Sigma_{x \in T} x\;(\bmod\;v)$, each $\mathcal{B}_\sigma$ is a $t-(v,t+1,1)$ packing.

Without restrictions on $v$, these $v$ packings need not have the same number of blocks. We now use the restriction that $(v,t+1) = 1$. Consider the orbits of $(t+1)$-subsets of $\mathbb{Z}_v$ under the cyclic action of $\mathbb{Z}_v$. When $S$ is a $(t+1)$-subset of $\mathbb{Z}_v$ with sum $\sigma$, let $S + \alpha$ be the subset of $\mathbb{Z}_v$ obtained by adding $\alpha$ (modulo $v$) to each element in $S$. Then the orbit containing $S$ is $\{S + \alpha: 0 \leq \alpha < v\}$. 
For $0 \leq \alpha < v$, the sum of $S + \alpha$ is $\sigma + (t+1)\alpha\;(\bmod\;v)$. Now if $S + \alpha$ and $S + \beta$ have the same sum modulo $v$, we have $(t+1)\alpha \equiv (t+1)\beta\;(\bmod\;v)$, which can happen only when $\alpha \equiv \beta\;(\bmod\;v)$. 
Hence every orbit contains exactly $v$ blocks, one in each of the $v$ classes. It follows that each $\mathcal{B}_\sigma$ contains $\frac{{v \choose t+1}}{v}$ blocks.

Now we prove statement (1). First we treat the cases when $\sigma \geq 0$. 
Choose $\sigma$ so that $0 \leq \sigma < {t+1 \choose 2}$ and consider the packing $D = ({\mathbb Z}_v,\mathcal{B}_\sigma)$. 
Suppose to the contrary that $S$ is a $(t+1)$-subset of $\mathbb{Z}_v$ with smallest sum $\tau < v + \sigma$. When $\tau \equiv \sigma\;(\bmod\;v)$ and $\tau < v + \sigma$, it must happen that $\sigma = \tau = \Sigma_{x \in S} x$. But $\Sigma_{x \in S} x \geq \Sigma_{i=0}^{t} i = {t+1 \choose 2} > \sigma$, which is a contradiction. 
Hence ${\sf MinSum}(D) \geq v + \sigma$. 
Because $\sigma \geq 0$, $tv + \sigma$ is the largest integer less than $(t+1)v$ that is congruent to $\sigma$ modulo $v$, and hence ${\sf MaxSum}(D) \leq tv + \sigma$.

Next we address the cases when $-{t+2 \choose 2} + 1 \leq \sigma < 0$. Let $\omega = v + \sigma$, and consider the packing $D=({\mathbb Z}_v,\mathcal{B}_\omega)$. Then ${\sf MinSum}(D) \geq\omega = v + \sigma < v$, because of the congruence requirement. 
For ${\sf MaxSum}(D)$,  suppose to the contrary that $S$ is a $(t+1)$-subset of $\mathbb{Z}_v$ with largest sum $\tau > tv + \sigma = (t-1)v + \omega$. Then $\tau = tv + \omega$. Now $\Sigma_{i=1}^{t+1} (v-i) = (t+1)v - {t+2 \choose 2} \geq \Sigma_{x \in S} x$. Hence $\omega \leq v - {t+2 \choose 2}$ so $\sigma \leq -{t+2 \choose 2}$, which is a contradiction.

Statements (2), (4), and (5) follow by taking $\sigma = {t+1 \choose 2} - 1$. Statement (3) follows by taking $\sigma = -{t+2 \choose 2} + 1$.
\end{proof}
\end{theorem}

Not surprisingly, the packings so produced contain large independent sets.
For example, when $\sigma=0$, the elements $\{0,\dots, \lfloor \frac{v}{t+1}\rfloor \}$ form an independent set.

Theorem \ref{mainpack} yields packings that are dense in the following sense.  
When an $S(t,t+1,v)$ exists, it has $ \frac{{v \choose t}}{{t+1 \choose t}}$ blocks; the packings considered have   a $\frac{v-t}{v}$ fraction of this number.  
Hence for fixed $t$ the fraction of $t$-sets left uncovered by the packing approaches 0 as $v \rightarrow \infty$.  
Moreover, the bounds established for dense $t$-$(v,t+1,1)$ packings on {\sf MinSum} and {\sf MaxSum} match the  values from Theorem \ref{basic} (which are best  for Steiner systems). 
On the other hand, as $v\rightarrow \infty$ and $t$ is fixed, the ratio of {\sf DiffSum} of the packing to the bound approaches 1, and the {\sf RatioSum} approaches its bound of $t-1$.
By generalizing to partial systems, Theorem \ref{mainpack} applies to all parameters that are large enough, whether or not an $S(t,t+1,v)$ exists.

Although Theorem~\ref{mainpack} establishes a {\sf DiffSum} of $(t-1)v$ for certain dense $t$-$(v,t+1,1)$ packings, one might hope to obtain a somewhat smaller {\sf DiffSum} when $t > 2$.  
Theorem \ref{fourpack} gives one result in this direction, producing a packing  that achieves a smaller {\sf DiffSum} than that of Theorem \ref{mainpack} when $t=3$, but is nearly as dense.   

\begin{theorem}\label{fourpack}
When $v > 18$ is even, there is a $3$-$(v,4,1)$ packing $D$ with $\frac{v-4}{v-1} \frac{\binom{v}{3}}{\binom{4}{3}}$ blocks,
having ${\sf MinSum}(D) = v+2$, ${\sf MaxSum}(D) = 3v-6$,  and hence ${\sf DiffSum}(D) = 2v-8$.
\end{theorem}
\begin{proof}
Write $v = 2s$.  We form $D$ on elements $\{0,\dots,2s-1\}$, with blocks 
\begin{enumerate}
\item $\{\{ a,b,c,s+d\} : 0 \leq a < b < c < s, 0 \leq d < s, a+b+c+d \equiv 2 \pmod{s}\}$, and
\item $\{\{ s+a,s+b,s+c,d\} : 0 \leq a < b < c < s, 0 \leq d < s, a+b+c+d \equiv s-6 \pmod{s}\}$.
\end{enumerate}
This forms a 3-$(v,4,1)$ packing with the specified number of blocks.  
Because $a + b + c + d \in \{s+2,2s+2,3s+2\}$,  blocks of the first class have sum in $\{2s+2,3s+2,4s+2\}$.
Similarly, because $a + b + c + d \in \{s-6,2s-6,3s-6\}$,  blocks of the second class have sum in $\{4s-6,5s-6,6s-6\}$.
Hence ${\sf MinSum}(D) = 2s+2 = v+2$ and  ${\sf MaxSum}(D) = 6s-6 = 3v-6$.
\end{proof}

\section{Sums and Steiner triple systems}\label{sumsts}

Despite the utility of dense packings in the intended applications, it remains desirable to employ a Steiner system when possible. 
In what follows, we extend Theorem~\ref{mainpack} to produce Steiner triple systems in which the sum metrics are close to optimal. 

Building on the construction in Theorem~\ref{mainpack},  Schreiber~\cite{schreiber1973covering} and Wilson~\cite{wilson1974some} demonstrate that for certain values of $v$, the packing can be completed to an STS$(v)$.  We provide a proof of their result, in order to examine the consequences for certain sums.
The construction relies on a number-theoretic property, which we state next without proof.

\begin{theorem}\label{minus2}{\rm \cite{schreiber1973covering}}
Every cycle in an abelian group $G$ of order $n$ contains twice an odd number of elements if and only if, for every prime divisor $p$ of $n$, the order of $-2\;(\bmod\;p)$ is singly even.
\end{theorem}

\begin{lemma}\label{STS-lemma}
Let $n \equiv 1,5 \pmod{6}$.  Every pair in $\{\{a, b\} : a,b \in {\mathbb Z}_n \setminus \{0\}, a \equiv -2b \pmod{n}\}$ has $(n+1)/2 \leq a+b \leq (n-1)/2 + n$.
\end{lemma}
\begin{proof}
Consider such a pair $a,b \in \mathbb{Z}_{n}$ with $b\equiv -2a\;(\bmod\;n)$. We examine two cases:

\noindent {Case 1: $1 \leq a \leq (n-1)/2$}. Then $b = n-2a$ and hence $(n+1)/2 \leq a+b \leq n-1$.

\noindent {Case 2: $(n+1)/2  \leq a \leq n-1$}: Then $b = 2n-2a$ and hence $n+1 \leq a+b \leq n+(n-1)/2$.
\end{proof}

\begin{theorem}
\label{special-SWC}
Whenever $v \equiv 1,3\;(\bmod\;6)$ and for every prime $p$ dividing $v-2$, it is the case that the order of $-2 \bmod p$ is singly even. Hence there is an STS$(v)$,  $D$, with 
${\sf MinSum}(D) \geq v-2$, 
${\sf MaxSum}(D)  \leq 2v + 2$,
 and hence ${\sf DiffSum}(D)  \leq v + 4$ and ${\sf RatioSum}(D)  \leq \frac{2v+2}{v -2}$.
\end{theorem}
\begin{proof}
Let $n = v-2$.  As in the proof of Theorem~\ref{mainpack}, construct the 2-$(v,3,1)$ packing ${\mathcal B}_0$ on ${\mathbb Z}_{n}$.  Each triple in this packing has sum $v-2$ or $2v-4$ at present.
The pairs left uncovered by any triple are $\mathcal{E}_0 = \{ \{ x, -2x \} : x \in {\mathbb Z}_{v-2} \setminus \{0\}\}$, each having sum between $(v-1)/2$ and  $(v-3)/2 + v-2$ by Lemma~\ref{STS-lemma}.

Because for every prime $p$ dividing $v-2$, it is the case that the order of $-2 \bmod p$ is singly even, Theorem~\ref{minus2} ensures that the pairs in $\mathcal{E}_0$ can be partitioned into two $1$-factors, $F_1$ and $F_2$, on ${\mathbb Z}_n \setminus \{0\}$.

To form an STS$(v)$ on ${\mathbb Z}_v$ with block set $\mathcal{C}$, we employ the mapping $\phi: \mathbb{Z}_{v-2} \mapsto \mathbb{Z}_v \setminus \{(v-1)/2,(v+1)/2\}$ defined by $\phi(x) = x$ when $0 \leq x \leq (v-3)/2$ and $\phi(x) = x+2$ when $(v-1)/2 \leq x < v-2$. Then $\mathcal{C}$ is formed as follows.
\begin{enumerate}
    \item[(1)] When $\{x,y,z\} \in \mathcal{B}_0$, place $\{\phi(x),\phi(y),\phi(z)\}$ in $\mathcal{C}$;
    \item[(2)] For $i = 1,2$, when $\{x,y\} \in F_i$, place $\{(v-3+2i)/2,\phi(x),\phi(y)\}$ in $\mathcal{C}$;
    \item[(3)] Place $\{0,(v-1)/2,(v+1)/2\}$ in $\mathcal{C}$.
\end{enumerate}
Triples of $\mathcal{B}_0$ have sum $v-2$ or $2v-4$, so triples of type $(1)$ in $\mathcal{C}$ have sum between $v-2$ and $v+2$, or between $2v-2$ and $2v+2$. 
A pair $\{x,y\} \in \mathcal{E}_0$ has  $(v-1)/2 \leq x+y \leq (v-1)/2 + (v-3)$. 
Applying $\phi$, we have $(v-1)/2 \leq \phi(x) + \phi(y) \leq (v-1)/2 + (v+1)$. 
Hence, each triple of type $(2)$ in $\mathcal{C}$ has sum at least $v-1$ and at most $2v+1$. 
Finally, the single type $(3)$ block has sum $v$.
\end{proof}

Unlike the point labelings in~\cite{dau2018maxminsum}, the labeling for the Schreiber-Wilson construction in Theorem~\ref{special-SWC} need not achieve the largest {\sf MinSum} or smallest {\sf MaxSum}. 
 Nevertheless it yields a substantial improvement on earlier constructions with respect to the {\sf DiffSum} and {\sf RatioSum}, within an additive constant of the best bound possible for the {\sf DiffSum}.
Unfortunately, Theorem \ref{special-SWC} requires that the order of $-2 \bmod p$ be singly even, and so applies to an infinite set of orders but not all admissible ones. 
We remedy this next, using a result from \cite{bryant2017steiner}, but obtaining slightly weaker bounds.

\begin{theorem}\label{SWC-gen}
Whenever $v \equiv 1,3\;(\bmod\;6)$, there is an STS$(v)$,  $D$, with 
${\sf MinSum}(D) \geq v-5$, 
${\sf MaxSum}(D)  \leq 2v + 2$,
 and hence ${\sf DiffSum}(D)  \leq v + 7$ and ${\sf RatioSum}(D)  \leq \frac{2v+2}{v -5}$.
\end{theorem}
\begin{proof}
Form $\mathcal{B}_0$ over $\mathbb{Z}_{v-2}$ as in the proof of Theorem \ref{special-SWC}.
Remove element $0$ as well as all triples $\{\{0,x,v-2-x\}: 1 \leq x \leq (v-3)/2\}$ to form $\mathcal{D}_0$. 
Let $\mathcal{E}_0$ be the set of pairs on $\mathbb{Z}_{v-2}$ not covered by a triple of $\mathcal{B}_0$. The pairs in $\mathcal{E}_0$, together with $\{\{x,v-2-x\}: 1 \leq x \leq (v-3)/2 \}$ form a 3-regular graph $G$ on $\mathbb{Z}_{v-2} \setminus \{0\}$. 
By  \cite[Lemma~9]{bryant2017steiner}, $G$ can be partitioned into three 1-factors, $F_1$, $F_2$, and $F_3$.

To form the STS$(v)$ on $\mathbb{Z}_v$ with block set $\mathcal{C}$, we employ the mapping $\psi: \mathbb{Z}_{v-2} \setminus \{0\} \mapsto \mathbb{Z}_v \setminus \{(v-3)/2,(v-1)/2,(v+1)/2\}$ defined by $\psi(x) = x-1$ when $1 \leq x \leq (v-3)/2$ and $\psi(x) = x + 2$ when $(v-1)/2 \leq x < v-2$. Then $\mathcal{C}$ is formed as follows.
\begin{enumerate}
    \item[(1)] When $\{x,y,z\} \in \mathcal{D}_0$, place $\{\psi(x), \psi(y), \psi(z)\}$ in $\mathcal{C}$;
    \item[(2)] For $i = 1,2,3$, when $\{x,y\} \in F_i$, place $\{(v-5+2i)/2,\psi(x),\psi(y)\}$ in $\mathcal{C}$;
    \item[(3)] Place $\{(v-3)/2,(v-1)/2,(v+1)/2\}$ in $\mathcal{C}$.
\end{enumerate}
Triples of $\mathcal{B}_0$ have sum $v - 2$ or $2v - 4$, so triples of type $(1)$ in $\mathcal{C}$ have sum between $v - 5$ and $v-2$, or between $2v-1$ and $2v+2$. 
By Lemma \ref{STS-lemma}, a pair $\{x,y\} \in \mathcal{E}_0$ has  $(v-1)/2 \leq x+y \leq (v-1)/2 + (v-3)$. 
Applying $\psi$, we have $(v-1)/2 - 2 \leq \psi(x) + \psi(y) \leq (v-1)/2 + (v+1)$. 
It follows that each triple of type $(2)$ in $\mathcal{C}$ has sum at least $v-4$ and at most $2v+1$. 
The block of type $(3)$ in $\mathcal{C}$ has sum $\frac{3v-3}{2}$.
\end{proof}

Although the bounds are slightly weaker, Theorem~\ref{SWC-gen} applies to all admissible orders for Steiner triple systems.
In conjunction with Theorem \ref{maxxm1}, for all $v \equiv 1,3 \pmod{6}$ with $v \geq 13$ one has $v+1 \leq {\sf DiffSum}(2,3,v) \leq v+7$ and $2+ \frac{1}{v} \leq {\sf RatioSum}(2,3,v) \leq  2 + \frac{12}{v-5}$.  
Using a simple computational search, we constructed $S(2,3,v)$s with specified {\sf MinSum}  and {\sf MaxSum}, as shown next:
\renewcommand{\tabcolsep}{2pt}

\begin{center}
\begin{tabular}{p{1in}cccc}
Order $v$ & {\sf MinSum} & {\sf MaxSum} & {\sf DiffSum} & {\sf RatioSum}\\
7 & $v-1$ & $2v-1$ & $v$ & $2+\frac{1}{v-1}$\\
9 & $v$ & $2v$ & $v$ & $2$ \\
13,15,19,21,25,27 & $v-1$ & $2v$ & $v+1$ & $2+\frac{2}{v-1}$ \\
7,15,19,21,27 & $v$ & $2v+1$ & $v+1$ & $2+\frac{1}{v}$ \\
13,25 & $v$ & $2v+2$ & $v+2$ & $2+\frac{2}{v}$ \\
\end{tabular}
\end{center}

It appears plausible that ${\sf DiffSum}(2,3,v) = v+1$ when $v \geq 13$.  
It also appears plausible that  ${\sf RatioSum}(2,3,v) \in \{ 2+\frac{1}{v}, 2+\frac{2}{v}\}$ for every $v\neq 9$, but there is insufficient data to speculate on when it takes the larger value and when the smaller. 

\section{Concluding remarks}

Because Theorem~\ref{mainpack} achieves a {\sf DiffSum} of $(t-1)v$ for dense $t$-$(v,t+1,1)$ packings, one might hope that this difference can be realized for $S(t,t+1,v)$ Steiner systems.
However, Theorem~\ref{diffp1} establishes that this does not happen when $t=2$ unless $v \in \{7,9\}$, although Theorem~\ref{SWC-gen} is within an additive constant.
The situation when $t=3$ appears to be  quite different. 
There is an $S(3,4,8)$ with blocks \[ \begin{array}{c} \{0127, 0136, 0145, 0235, 0246, 0347, 0567, \\ 1234, 1256, 1357, 1467, 2367, 2457, 3456\}, \end{array}  \] having {\sf MinSum} 10 and {\sf MaxSum} 18.
Adapting the construction in~\cite{DoyenV71,PhelpsRosa80}, one can produce an $S(3,4,v)$ with {\sf MinSum} $v+2$,  {\sf MaxSum} $3v-6$, and hence {\sf DiffSum} $2v-8$ whenever $v$ is a power of 2.  
In these cases, the upper bound on the {\sf MinSum} and the lower bound on the {\sf MaxSum} from Theorem~\ref{basic} are met simultaneously.
We do not expect this to happen for all orders, because the smallest {\sf DiffSum} for an $S(3,4,v)$ when $v \in \{10,14\}$ appears to arise from  systems with {\sf MinSum} $v+1$ and {\sf MaxSum} $3v-5$.
It may happen that for every admissible $v$, an $S(3,4,v)$ with {\sf DiffSum} strictly smaller than $2v$ exists.
If so, completing the packing from Theorem~\ref{mainpack} could not yield the smallest {\sf DiffSum}.
Nevertheless,  the structure of independent sets must underlie  appropriate constructions. 

\section*{Acknowledgements}
The work was supported by  NSF grants CCF 1816913 (CJC) and 1814298 (OM).


\end{document}